\documentclass{article}
\usepackage{ijcai24}

\usepackage{times}
\usepackage{soul}
\usepackage{url}
\usepackage[hidelinks]{hyperref}
\usepackage[utf8]{inputenc}
\usepackage[small]{caption}
\usepackage{graphicx}
\usepackage{amsmath}
\usepackage{amsthm}
\usepackage{booktabs}
\usepackage{subfigure}
\usepackage{multirow}
\usepackage[ruled,vlined]{algorithm2e}
\usepackage[switch]{lineno}

\newtheorem{definition}{Definition}

\newtheorem{lemma}{Lemma}

\usepackage{xcolor}

\usepackage{xspace}
\makeatletter
\DeclareRobustCommand\onedot{\futurelet\@let@token\@onedot}
\def\@onedot{\ifx\@let@token.\else.\null\fi\xspace}
\def\eg{\emph{e.g}\onedot} 
\def\ie{\emph{i.e}\onedot}

\makeatother

\title{Two New Upper Bounds for the Maximum $k$-plex Problem}

\author{
    Jiongzhi Zheng\equalcontrib \and  
    Mingming Jin\equalcontrib \and    
    Kun He\thanks{Corresponding author.}
    \affiliations
    School of Computer Science and Technology, Huazhong University of Science and Technology, China
    \\
    \emails
    brooklet60@hust.edu.cn
}

\begin{document}

\maketitle

\begin{abstract}

   A $k$-plex in a graph is a vertex set where each vertex is non-adjacent to at most $k$ vertices (including itself) in this set, and the Maximum $k$-plex Problem (MKP) is to find the largest $k$-plex in the graph. 
   As a practical NP-hard problem, MKP has many important real-world applications, such as the analysis of various complex networks. Branch-and-bound (BnB) algorithms are a type of well-studied and effective exact algorithms for MKP. Recent BnB MKP algorithms involve two kinds of upper bounds based on graph coloring and partition, respectively, that work in different perspectives and thus are complementary with each other. In this paper, we first propose a new coloring-based upper bound, termed Relaxed Graph Color Bound (RelaxGCB), that significantly improves the previous coloring-based upper bound. We further propose another new upper bound, termed RelaxPUB, that incorporates RelaxGCB and a partition-based upper bound in a novel way, making use of their complementarity. We apply RelaxGCB and RelaxPUB to state-of-the-art BnB MKP algorithms and produce eight new algorithms. 
   Extensive experiments using diverse $k$ values on hundreds of instances based on dense and massive sparse graphs 
   demonstrate the excellent performance and robustness of our proposed methods. 
\end{abstract}

\section{Introduction}
\label{sec-Intro}
Given an undirected graph $G = (V, E)$, a clique is a set of vertices that are pairwise adjacent, and a $k$-plex~\cite{SF78} is a set of vertices $S \subseteq V$ where each vertex $v \in S$ is non-adjacent to at most $k$ vertices (including $v$ itself) in $S$. The Maximum Clique Problem (MCP) is to find the largest clique in $G$, while the Maximum $k$-plex Problem (MKP) is to find the largest $k$-plex in $G$.

MCP is a famous and fundamental NP-hard problem, and the clique model has been widely investigated in the past decades. However, in many real-world applications, such as social network mining~\cite{SF78,PYB13,CMS+18,ZCZ20,WangLHH23TKDD} 
and biological network analysis~\cite{GJM+20}, dense subgraphs need not to be restrictive 
cliques but allow missing a few connections. Therefore, investigating relaxation clique structures like $k$-plex is significant, and studies related to $k$-plex have sustainably grown in recent decades~\cite{BBH11,McCloskyH12,Berlowitz2015,ConteFMPT2017,WangZXK22}.

Many efficient exact methods for the NP-hard MKP have been proposed~\cite{XLD+17,GCY+18,ZHX+21,JZX+21,Chang2022kPlexS,wang2023fast,jiang2023Dise}, resulting in various effective techniques, such as reduction rules, upper bounds, inprocessing methods, etc. 
Most of these studies follow the branch-and-bound (BnB) framework~\cite{LawlerW66,LiQ10,McCreeshPT17}, and their performance heavily depends on the quality of the upper bounds.



A BnB MKP algorithm usually maintains the current growing partial $k$-plex $S \subseteq V$ and the corresponding candidate vertex set $C \subseteq V \backslash S$. Methods for calculating the upper bound on the number of vertices that $C$ can provide for $S$ in existing BnB MKP algorithms can be divided into two categories. The first calculates the upper bound by considering the connectivity between vertices in $C$ only, 
such as the graph color bound (GCB) proposed in the Maplex algorithm~\cite{ZHX+21}. The second considers the connectivity between vertices in $C$ and vertices in $S$, including the partition-based upper bounds (PUB) proposed in the KpLeX~\cite{JZX+21} 
algorithm and also used in the kPlexS~\cite{Chang2022kPlexS} and KPLEX~\cite{wang2023fast} algorithms.

In this work, we observe that 
the upper bounds of the above algorithms are still not very tight. For a graph $G$, an independent set $I$ is a subset of $V$ where any two vertices are non-adjacent. Graph coloring assigns a color to each vertex such that adjacent vertices are in different colors, which is widely used for finding independent sets in graphs. 
GCB~\cite{ZHX+21} claims that an independent set $I \subseteq C$ can provide at most $\min\{|I|,k\}$ vertices for $S$, which actually ignores the connectivity between vertices in $I$ and vertices in $S$. 
While PUB~\cite{JZX+21} simply regards $C$ as a clique. Also, due to different motivations of the two kinds of upper bounds, they show complementary performance in various instances, as indicated in our follow-up examples and experiments.

To this end, we propose a new upper bound based on graph coloring called Relaxed Graph Color Bound (RelaxGCB). 
RelaxGCB first calculates an upper bound for each independent set $I \subseteq C$ that is strictly no worse than GCB by considering the connectivity between not only vertices in $I$ themselves but also vertices in $I$ and vertices in $S$. Furthermore, RelaxGCB relaxes the restrictive structure of independent sets, allowing to add some extra vertices to a maximal independent set (\ie, not contained by any other independent set) $I \subseteq C$ without increasing the upper bound.

Based on our observation that the coloring-based and partition-based upper bounds are complementary, we propose another new upper bound called RelaxPUB. 
RelaxPUB combines our RelaxGCB with a refined PUB called DisePUB proposed in DiseMKP~\cite{jiang2023Dise}. Different from common methods for combining various upper bounds that sequentially calculate them until the branch can be pruned or cannot be pruned by any upper bound, RelaxPUB combines RelaxGCB and DisePUB in a novel and compact way. 
When calculating the upper bound of the number of vertices that $C$ can provide for $S$, both of them iteratively
extracts a subset $I \subseteq C$ from $C$, calculating the upper bound of the number of vertices that $I$ can provide for $S$ and accumulating the upper bounds. 
In each iteration, RelaxPUB uses RelaxGCB and DisePUB to respectively extract a subset from $C$ and selects the better one, 
and repeats such a process until $C$ is empty.



We evaluate our proposed two upper bounds by applying them to state-of-the-art BnB MKP algorithms, including Maplex, kPlexS, DiseMKP, and KPLEX. Among them, Maplex only applies coloring-based upper bound, \ie, GCB, and the others only apply PUB. We replace their original upper bounds with our RelaxGCB and RelaxPUB. Extensive experiments show that in both dense and massive sparse graphs using various $k$ values, RelaxGCB is a significant improvement over the GCB, and RelaxPUB can significantly improve the baseline algorithms, indicating the excellent and generic performance of our methods.

\section{Preliminaries}
\label{sec-Pre}


\subsection{Definitions}
Given an undirected graph $G = (V, E)$, where $V$ is the vertex set and $E$ the edge set, the density of $G$ is $2|E| / (|V|(|V|-1))$, we denote $N(v)$ as the set of vertices adjacent to $v$, 
which are also called the neighbors of $v$. 
Given a vertex set $S \subseteq V$, we denote $G[S]$ as the subgraph induced by $S$. Given an integer $k$, $S \subseteq V$ is a $k$-plex if each vertex $v \in S$ satisfies that $|S \backslash N(v)| \leq k$. 

For a growing partial $k$-plex $S$, we define $\omega_k(G,S)$ as the size of the maximum $k$-plex that includes all vertices in $S$, and $\delta(S,v) = |S \backslash N(v)|$ as the number of non-neighbors of vertex $v$ in $S$. Given an integer $k$, we further define $\delta_k^{-}(S,v) = k - \delta(S,v)$ to facilitate our algorithm description. 
If $v \in S$, $\delta_k^{-}(S,v)$ indicates the maximum number of non-adjacent vertices of $v$ that can be added to $S$. Otherwise, it indicates that, including $v$ itself, the maximum number of its non-adjacent vertices that can be added to $S$.

\subsection{Framework of BnB MKP Algorithms}
During the course of a general BnB MKP algorithm, a lower bound $lb$ on the size of the maximum $k$-plex is maintained, which is usually initialized by some heuristic algorithms~\cite{ZHX+21,JZX+21,Chang2022kPlexS}, and is updated once a larger $k$-plex is found. 

A general BnB MKP algorithm usually contains a preprocessing stage and a BnB search stage. 
During the preprocessing, the algorithm uses some reduction rules~\cite{GCY+18,ZHX+21,Chang2022kPlexS} to remove vertices that are impossible to belong to a $k$-plex of size larger than $lb$. In the BnB search stage, the algorithm traverses the search tree to find the optimal solution. During the search, the algorithm always maintains two vertex sets, the current growing partial $k$-plex $S$, and its corresponding candidate set $C$ containing vertices that might be added to $S$. Once the algorithm selects a branching vertex $v$ to be added to $S$ from $C$, it calculates an upper bound $ub$ on the size of the maximum $k$-plex that can be extended from $S \cup \{v\}$, and the branch of adding $v$ to $S$ will be pruned if $ub \leq lb$. 

\section{The RelaxGCB Bound}
Given a growing partial $k$-plex $S$ and the corresponding candidate vertex set $C$, the graph color bound (GCB) proposed in Maplex~\cite{ZHX+21} claims that an independent set $I \subseteq C$ can provide at most $\min\{|I|,k\}$ vertices for $S$. As introduced in Section~\ref{sec-Intro}, our proposed Relaxed Graph Color Bound (RelaxGCB) improves GCB from two aspects, \ie, calculating a tighter bound for each independent set $I \subseteq C$ and allowing add extra vertices to a maximal independent set without changing the upper bound.


In the following, we first introduce our two improvements and provide an example for illustration, then present our RelaxColoring algorithm for calculating the RelaxGCB bound.


\subsection{A Tighter Upper Bound for Independent Sets}
\label{sec-RelaxGCB1}
Since vertices in the candidate set $C$ might be non-adjacent to some vertices in the growing partial $k$-plex $S$, an independent set $I \subseteq C$ actually cannot provide $k$ vertices for $S$ sometimes even when $|I| > k$. We introduce a Tighter Independent Set Upper Bound (TISUB) on the number of vertices that an independent set $I \subseteq C$ can provide for $S$.



\begin{lemma}[TISUB]
\label{lemma-RelaxGCB1}
    Suppose $I = \{v_1, v_2, \cdots, v_{|I|}\} \subseteq C$ is an independent set and $\delta_k^{-}(S,v_1) \geq \delta_k^{-}(S,v_2) \geq \cdots \geq \delta_k^{-}(S,v_{|I|})$, $\mathop{\max}\{i | \delta_k^{-}(S,v_i) \geq i\}$ is an upper bound of the number of vertices that $I$ can provide for $S$.
\end{lemma}

\begin{proof}
    Firstly, ignoring the constraint of at most $k$ non-neighbors of vertices in $S$, 
    $v_1, v_2, \cdots, v_{|I|}$ is one of the best orders for adding vertices in $I$ to $S$ to obtain the largest $k$-plex in $G[S \cup I]$, because the more non-neighbors in $S$ (as indicated by $\delta(S,v)$), the easier it is for vertices to violate the constraint. 
    Secondly, suppose vertices $v_1, \cdots, v_i$ are going to be added to $S$, further adding $v_{i+1}$ to $S$ leads to $\delta(S,v_{i+1}) + i + 1$ non-neighbors of $v_{i+1}$ in $S$ (including $v_{i+1}$ itself). Therefore, only vertices $v_i \in I$ with $\delta(S,v_i) + i \leq k$, \ie, $\delta_k^{-}(S,v_i) \geq i$, can be added to $S$, and $I$ can provide at most $\mathop{\max}\{i | \delta_k^{-}(S,v_i) \geq i\}$ vertices for $S$. 
\end{proof}


For convenience, in the rest of this paper, we regard the vertices in any independent set $I \subseteq C$, \ie, $\{v_1,v_2,\cdots,v_{|I|}\}$, as sorted in 
non-ascending order of their $\delta_k^{-}(S,v)$ values.
We further define \textit{TISUB}$(I,S) = \mathop{\max}\{i | \delta_k^{-}(S,v_i) \geq i\}$ as the upper bound calculated by TISUB on the number of vertices that $I$ can provide for $S$. Note that the value of \textit{TISUB}$(I,S)$ is obviously bounded by $|I|$ since $i \leq |I|$, which eliminates the need for term $|I|$ in TISUB. Moreover, since $\delta(S,v) \geq 0$, $\delta_k^{-}(S,v) \leq k$ holds, and \textit{TISUB}$(I,S)$ is also bounded by $k$. Therefore, TISUB is strictly never worse than GCB (\ie, $\min\{|I|, k\}$).


\subsection{Relax the Independent Sets}
Since the relaxation property of $k$-plex over clique, an independent set $I$ in the candidate set $C$ can usually provide more than one vertices for the growing the partial $k$-plex $S$, and the restriction of independent sets can also be relaxed to contain more vertices. 

In the following, we define two kinds of vertices and then introduce two different rules for relaxing the restriction of independent sets and making maximal independent sets contain extra vertices without increasing their TISUB. 

\begin{definition}[Conflict Vertex]
\label{definition-CV}
    Given a vertex set $I$, we denote vertices $v \in I$ that are adjacent to at least one vertex in $I$ as conflict vertices.
\end{definition}

\begin{definition}[Loose Vertex]
\label{definition-SV}
    Given a $k$-plex $S$ and a vertex set $I \subseteq C$, suppose $UB$ is an upper bound of the number of vertices that $I$ can provide for $S$, we denote each vertex $v \in I$ with $\delta_k^{-}(S,v) > UB$ as a loose vertex.
\end{definition}





\noindent \textbf{Rule 1.} Suppose $UB$ is an upper bound of the number of vertices that a vertex set $I \subseteq C$ can provide for $S$. It is allowed to add vertex $v$ to $I$ if the number of vertices that are loose or conflict in $I \cup \{v\}$ is no more than $UB$.

\begin{lemma}
\label{lemma-R1}
    After adding any vertex $v$ to $I \subseteq C$ according to Rule 1, $UB$ is still an upper bound of the number of vertices that $I' = I \cup \{v\}$ can provide for $S$.
\end{lemma}

\begin{proof}
    On one hand, if adding a vertex $v \in I'$ that is neither \textit{conflict} nor \textit{loose} to $S$, then at most $\delta_k^{-}(S,v) - 1 < UB$ other vertices in $I'$ can be added to $S$. On the other hand, 
    by Rule 1, we require the number of \textit{conflict} or \textit{loose} vertices in $I'$ to be no more than $UB$. Therefore, at most $UB$ vertices in $I'$ can be added to $S$.
\end{proof}

\noindent \textbf{Rule 2.} Suppose $UB$ is an upper bound of the number of vertices that a vertex set $I \subseteq C$ can provide for $S$. It is allowed to add vertex $v$ to $I$ if $v$ is adjacent to at most $UB - \delta_k^{-}(S,v)$ vertices in $I$.

\begin{lemma}
\label{lemma-R2}
    After adding any vertex $v$ to $I \subseteq C$ according to Rule 2, $UB$ is still an upper bound of the number of vertices that $I' = I \cup \{v\}$ can provide for $S$.
\end{lemma}

\begin{proof}
    On one hand, if adding $v$ to $S$, at most $\delta_k^{-}(S,v) - 1$ other vertices that are non-adjacent to $v$ in $I'$ can be added to $S$. Since $v$ is adjacent to at most $UB - \delta_k^{-}(S,v)$ vertices in $I'$, thus after adding $v$ to $S$, $I'$ can still provide at most $UB - 1$ vertices for $S$. On the other hand, if not adding $v$ to $S$, $I'$ itself can only provide at most $UB$ vertices for $S$.
\end{proof}

Given a maximal independent set $I \subseteq C$, both Rule 1 and Rule 2 can add extra vertices to $I$ without increasing its TISUB. Actually, Rule 1 allows us to add finite (at most \textit{TISUB}$(I,S)$ - 1) \textit{conflict} vertices to $I$, and Rule 2 can be repeatedly used to add any vertex satisfying the rule to $I$.



\begin{figure}[t]
\centering
\includegraphics[width=0.45\columnwidth]{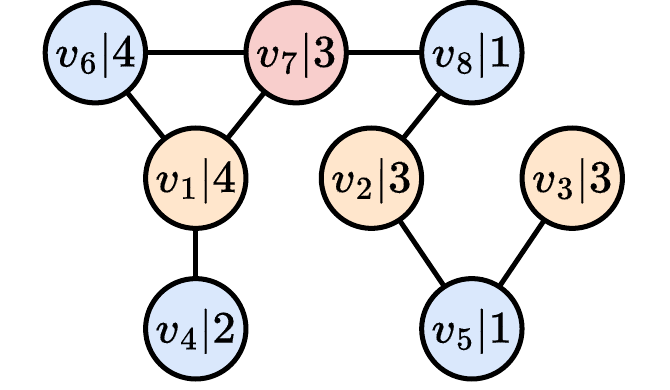}
\caption{An example for comparing the upper bounds.\vspace{-1em}}
\label{fig-bound}
\end{figure}

\subsection{An Example for Illustration}
We provide an example in Figure~\ref{fig-bound} to show how the upper bounds, including GCB, TISUB, and RelaxGCB, are calculated and how the two rules are used. 
Figure~\ref{fig-bound} illustrates a subgraph of $G$ induced by the candidate set $C = \{v_1, v_2, \cdots, v_8\}$, \ie, $G[C]$, of a 4-plex $S$. To simplify the figure, we hide the 4-plex $S$ and only depict the candidate vertices. Vertex $v_i|t$ in Figure~\ref{fig-bound} identifies a vertex $v_i \in C$ with $\delta_k^{-}(S,v_i) = t$.

Suppose we sequentially color vertices $v_1, v_2, \cdots, v_8$ under the constraint that adjacent vertices cannot be in the same color, $C$ can be partitioned into 3 independent sets, $I_1 = \{v_1,v_2,v_3\}$, $I_2 = \{v_4,v_5,v_6,v_8\}$ and $I_3 = \{v_7\}$, as indicated by the colors of the vertices. The GCB of $\omega_4(G,S)$ is $|S| + \sum_{i=1}^3{\min\{|I_i|,4\}} = |S| + 3 + 4 + 1 = |S| + 8$. The TISUB of $\omega_4(G,S)$ is $|S| + \sum_{i=1}^3{\textit{TISUB}(I_i,S)} = |S| + 3 + 2 + 1 = |S| + 6$.

Then, let us use Rule 1 to make independent set $I_1$ contain more vertices. For $I_1$, since $\textit{TISUB}(I_1,S) = 3$, there is only one \textit{loose} vertex $v_1$ in $I_1$. By applying Rule 1, we can add vertices $v_6$ and $v_7$ to $I_1$ without increasing the upper bound of $\omega_4(G[S \cup I_1],S)$, since there are only 3 \textit{loose} or \textit{conflict} vertices, \ie, $v_1,v_6,v_7$, in $I_1 \cup \{v_6,v_7\}$. After the operation, $C$ is partitioned into two sets, $I_5 = I_1 \cup \{v_6,v_7\}$ and $I_6 = \{v_4,v_5,v_8\}$. The new upper bound of $\omega_4(G,S)$ is $|S| + \textit{TISUB}(I_5,S) + \textit{TISUB}(I_6,S) = |S| + 3 + 1 = |S| + 4$.

Finally, let us use Rule 2 to further make set $I_5$ contain more vertices. According to Rule 2, all vertices in $I_6$ can be added to $I_5$ without increasing the upper bound of $\omega_4(G[S \cup I_5],S)$. After the operation, the final RelaxGCB of $\omega_4(G,S)$ is $|S| + \textit{TISUB}(I_5,S) = |S| + 3$.

\subsection{The RelaxColoring Algorithm}
\label{sec-RelaxColoring}
This subsection introduces our proposed RelaxColoring algorithm for calculating the proposed RelaxGCB, as summarized in Algorithm~\ref{alg-RelaxColoring}. 
The algorithm first uses $|S|$ to initialize the upper bound $UB$ (line 1), and then repeatedly uses the TryColor() function to extract a subset $I \subseteq C$ and calculate the upper bound on the number of vertices that $I$ can provide for $S$, \ie, $ub$ (line 3) until $C = \emptyset$ (line 2). After each execution of function TryColor(), the candidate set $C$ and upper bound $UB$ are both updated (line 4).

\begin{algorithm}[t]
\fontsize{10.1pt}{15}
\caption{RelaxColoring$(G,k,S,C)$}
\label{alg-RelaxColoring}
\LinesNumbered 
\KwIn{A graph $G = (V,E)$, an integer $k$, the current partial $k$-plex $S$, the candidate set $C$}
\KwOut{RelaxGCB of $\omega_k(G,S)$}
initialize the upper bound $UB \leftarrow |S|$\;
\While{$C \neq \emptyset$}{
$\{I,ub\} \leftarrow$ TryColor$(G,k,S,C)$\;
$C \leftarrow C \backslash I$, $UB \leftarrow UB + ub$\;
}
\textbf{return} $UB$\;
\end{algorithm}

\begin{algorithm}[t]
\fontsize{10.1pt}{15}
\caption{TryColor$(G,k,S,C)$}
\label{alg-TryColor}
\LinesNumbered 
\KwIn{A graph $G = (V,E)$, an integer $k$, the current partial $k$-plex $S$, the candidate set $C$}
\KwOut{A vertex set $I$, an upper bound $ub$ of the number of vertices that $I$ can provide for $S$}
initialize $I \leftarrow \emptyset$\;
\For{\rm{\textbf{each}} vertex $v \in C$}{
\lIf{$N(v) \cap I = \emptyset$}{$I \leftarrow I \cup \{v\}$}
}
$ub \leftarrow$ \textit{TISUB}$(I,S)$\;
initialize the set of \textit{loose} or \textit{conflict} vertices $LC \leftarrow \{v \in I | \delta_k^{-}(S,v) > ub\}$\;
\If{$|LC| < ub$}{
\For{\rm{\textbf{each}} vertex $v \in C \backslash I$}{
$CV \leftarrow \{v\} \cup \{N(v) \cap I \backslash LC\}$\;
\If{$|LC| + |CV| \leq ub$}{
$I \leftarrow I \cup \{v\}$\;
$LC \leftarrow LC \cup CV$\;
\lIf{$|LC| = ub$}{\textbf{break}}
}
}
}
\For{\rm{\textbf{each}} vertex $v \in C \backslash I \wedge \delta_k^{-}(S,v) < ub$}{
\If{$|N(v) \cap I| \leq ub - \delta_k^{-}(S,v)$}{$I \leftarrow I \cup \{v\}$\;}
}
\textbf{return} $\{I,ub\}$\;
\end{algorithm}

Function TryColor() is summarized in Algorithm~\ref{alg-TryColor}, which first finds a maximal independent set $I \subseteq C$ (lines 1-3) and calculates its TISUB (line 4). Then, the algorithm initializes the set of \textit{loose} or \textit{conflict} vertices $LC$ (line 5) and tries to add as many vertices as possible to $I$ according to Rule 1 (lines 6-12). Once trying to add each vertex $v$, the algorithm uses $CV$ to denote the extra \textit{conflict} vertices caused by adding $v$ to $I$ (line 8). Since $I$ is a maximal independent set in $C$, adding any vertex $v$ to $I$ increases at least one \textit{conflict} vertices, \ie, $v$ itself (line 8). Thus, the utilization of Rule 1 can be terminated when $|LC| \geq ub$ (lines 6 and 12). Finally, the algorithm applies Rule 2 to further add vertices to $I$ (lines 13-15). Since for each vertex $v \in C \backslash I$, $|N(v) \cap I| > 0$ holds, only vertex $v \in C \backslash I$ with $\delta_k^{-}(S,v) < ub$ can be added to $I$ according to Rule 2 (line 13).

The time complexities of RelaxColoring algorithm and TryColor function are $O(|C|^2 \times T)$ and $O(|C| \times T)$, respectively, where $O(T)$ is the time complexity of the intersection operation between $N(v)$ and $I$ (or $I \backslash LC$) used in lines 3, 8, and 14 in Algorithm~\ref{alg-TryColor}. Actually, $O(T)$ is bounded by $O(|V|)$ and much smaller than $O(|V|)$ by applying the bitset encoding method~\cite{SRJ11}.

\section{The RelaxPUB Bound}
Motivated by the complementarity of the coloring-based and partition-based upper bounds, 
we propose to combine RelaxGCB with the newest PUB, DisePUB~\cite{jiang2023Dise}, and propose a better and generic upper bound for MKP. In this section, we first introduce DisePUB, then provide two examples to illustrate the complementarity of the coloring-based and partition-based upper bounds, and finally present our new upper bound, RelaxPUB.

\subsection{Revisiting DisePUB}
Given a growing partial $k$-plex $S$ and the corresponding candidate set $C$, for each vertex $v \in S$, DisePUB claims that a subset $I \subseteq C$ can provide at most $\min\{|I|,\delta_k^{-}(S,v)\}$ vertices for $S$ if $N(v) \cap I = \emptyset$. 
Given a vertex $v \in S$, let $I = C \backslash N(v)$ and $ub = \min\{|I|,\delta_k^{-}(S,v)\}$, DisePUB defines a metric for $I$, \ie, $dise(I) = |I|/ub$, to evaluate the extraction of vertex set $I$. The larger the value of $dise(I)$, the more vertices that can be extracted from $C$ and the fewer increments on the upper bound of $\omega_k(G,S)$.

In each step, DisePUB traverses each vertex $v \in S$ with $\delta_k^{-}(S,v) > 0$ and selects the corresponding set $I = C \backslash N(v)$ with the largest value of $dise(I)$. Ties are broken by preferring larger extractions. We use function SelectPartition() to describe the selection, which is shown in Algorithm~\ref{alg-SelectPartition}. Then, DisePUB extracts $C \backslash N(v)$ from $C$ and increases the upper bound of $\omega_k(G,S)$ by $\min\{|C \backslash N(v)|,\delta_k^{-}(S,v)\}$.

DisePUB repeats the above process until vertices remaining in $C$ are adjacent to all vertices in $S$. DisePUB denotes the set of remaining vertices in $C$ as $\pi_0$ and finally increases the upper bound of $\omega_k(G,S)$ by $|\pi_0|$. 

\begin{algorithm}[t]
\fontsize{10.1pt}{15}
\caption{SelectPartition$(G,k,S,C)$}
\label{alg-SelectPartition}
\LinesNumbered 
\KwIn{A graph $G = (V,E)$, an integer $k$, the current partial $k$-plex $S$, the candidate set $C$}
\KwOut{A vertex set $I$, an upper bound $ub$ of the number of vertices that $I$ can provide for $S$}
initialize $dise^* \leftarrow 0, ub^* \leftarrow 0, I^* \leftarrow \emptyset$\;
\For{\rm{\textbf{each}} vertex $v \in S \wedge \delta_k^{-}(S,v) > 0$}{
$I \leftarrow C \backslash N(v)$\;
$ub \leftarrow \min\{|I|,\delta_k^{-}(S,v)\}$\;
\If{$|I|/ub > dise^* \vee (|I|/ub = dise^* \wedge |I| > |I^*|)$}{
$dise^* \leftarrow |I|/ub, ub^* \leftarrow ub, I^* \leftarrow I$\;
}
}
\textbf{return} $\{I^*,ub^*\}$\;
\end{algorithm}

\subsection{Complementarity of GCB and PUB}

\begin{figure}[t]
\centering
\subfigure[GCB prevails]{
\includegraphics[width=0.4\columnwidth]{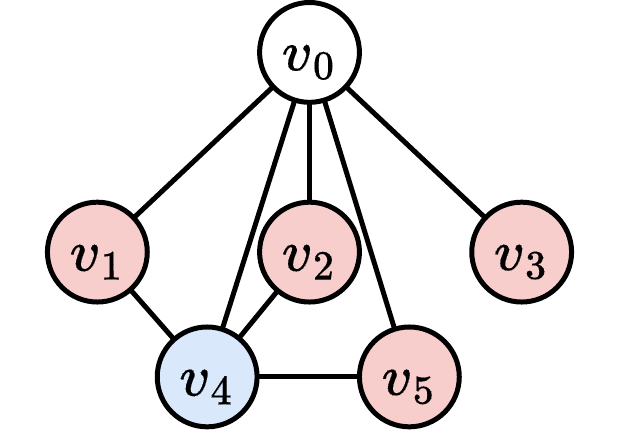}
\label{Complementary1}\hspace{1em}}
\subfigure[PUB prevails]{
\includegraphics[width=0.4\columnwidth]{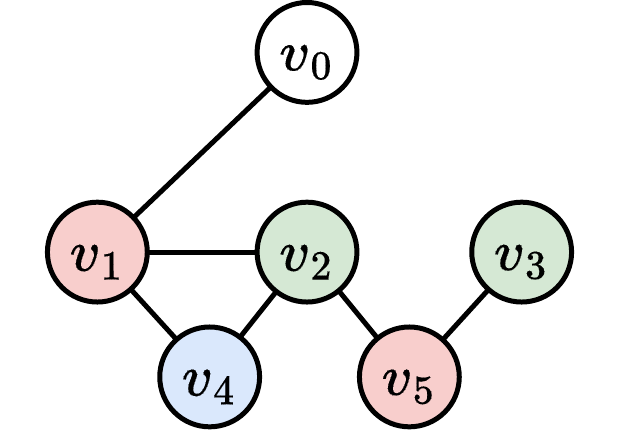}
\label{Complementary2}}
\vspace{-1em}\caption{Two examples for demonstrating the complementarity.\vspace{-0.5em}}
\label{fig-complementary}
\end{figure}
To better illustrate the complementarity of the coloring-based and partition-based upper bounds (\ie, GCB and PUB), we provide two examples in Figure~\ref{fig-complementary}, where the growing 2-plex $S$ contains only one vertex $v_0$ and its corresponding candidate set $C = \{v_1,v_2,v_3,v_4,v_5\}$. 

In Figure~\ref{Complementary1}, the GCB is tighter than the PUB. The vertices in $C$ are all adjacent to $v_0$, which means the vertices in $C$ are all in $\pi_0$. Thus, the PUB is $|S| + |\pi_0| = 6$. While by coloring the vertices in $C$, it can be partitioned into 2 independent sets $I_1 = \{v_1,v_2,v_3,v_5\}$ and $I_2 = \{v_4\}$, and the GCB is $|S|+\sum_{i=1}^2{\min\{|I_i|,2\}} = 4$. In contrast, the PUB is tighter than the GCB in Figure~\ref{Complementary2}, where $C$ can be partitioned into 3 independent sets $I_1 = \{v_1,v_5\}$, $I_2 = \{v_2,v_3\}$, and $I_3 = \{v_4\}$. Thus, the GCB is $|S|+\sum_{i=1}^3{\min\{|I_i|,2\}} = 6$. While 
vertices in $C$ except $v_1$ are non-adjacent to $v_0$, $\pi_0 = \{v_1\}$, thus the PUB is $|S| + |\pi_0| +\delta_2^{-}(S,v_0) = 3$.

\begin{algorithm}[t]
\fontsize{10.1pt}{15}
\caption{SelectUB$(G,k,S,C)$}
\label{alg-Seesawing}
\LinesNumbered 
\KwIn{A graph $G = (V,E)$, an integer $k$, the current partial $k$-plex $S$, the candidate set $C$}
\KwOut{RelaxPUB of $\omega_k(G,S)$}
initialize the upper bound $UB \leftarrow |S|$\;
\While{$C \neq \emptyset$}{
$\{I_{C},ub_{C}\} \leftarrow$ TryColor$(G,k,S,C)$\;
$\{I_{P},ub_{P}\} \leftarrow$ SelectPartition$(G,k,S,C)$\;
\eIf{$|I_{C}|/ub_{C} > |I_{P}|/ub_{P} \vee (|I_{C}|/ub_{C} = |I_{P}|/ub_{P} \wedge |I_{C}| > |I_R|)$}{
$C \leftarrow C \backslash I_{C}$, $UB \leftarrow UB + ub_{C}$\;
}{
$C \leftarrow C \backslash I_{P}$, $UB \leftarrow UB + ub_{P}$\;
}
}

\textbf{return} $UB$\;
\end{algorithm}

\subsection{Combining RelaxGCB and DisePUB}
Both RelaxGCB and DisePUB extract a subset from $C$ and accumulate the upper bound of $\omega_k(G,S)$. The $dise$ metric in DisePUB can also be used for the vertex set returned by TryColor(). RelaxPUB combines RelaxGCB and DisePUB by using them to select a promising extraction in each step.


We propose an algorithm called SelectUB for calculating the RelaxPUB of $\omega_k(G,S)$, which is presented in Algorithm~\ref{alg-Seesawing}. The algorithm calls TryColor() and SelectPartition() in each step and figures out whose returned vertex set is better according to the $dise$ metric. Ties are broken by preferring larger extraction. Once a better extraction is selected, The algorithm updates the candidate set $C$ and accumulates the upper bound of $\omega_k(G,S)$.

The time complexities of functions TryColor() and SelectPartition() are $O(|C| \times T)$ and $O(|C| \times |S|)$~\cite{jiang2023Dise}, respectively, where $O(T)$ is much smaller than $O(|V|)$ as referred to Section~\ref{sec-RelaxColoring}. The time complexity of the SelectUB algorithm is $O(|C|^2 \times (|S| + T))$.


\section{Experimental Results}
This section presents experimental results to evaluate the performance of the proposed two new upper bounds, RelaxGCB and RelaxPUB. We select state-of-the-art BnB MKP algorithms as the baselines, including Maplex\footnote{https://github.com/ini111/Maplex}~\cite{ZHX+21}, kPlexS\footnote{https://lijunchang.github.io/Maximum-kPlex}~\cite{Chang2022kPlexS}, DiseMKP\footnote{https://github.com/huajiang-ynu/ijcai23-kpx}~\cite{jiang2023Dise}, an improvement version of KpLeX~\cite{JZX+21}, and KPLEX\footnote{https://github.com/joey001/kplex\_degen\_gap}~\cite{wang2023fast}.

We replace the original upper bounds in the baselines with our RelaxGCB and RelaxPUB and conduct eight new BnB algorithms. The new algorithms based on Maplex with our upper bounds are denoted as RelaxGCB-Maplex and RelaxPUB-Maplex, respectively, and so on. 



\begin{figure*}[!t]
\centering
\subfigure[Comparison with Maplex]{
\includegraphics[width=0.49\columnwidth]{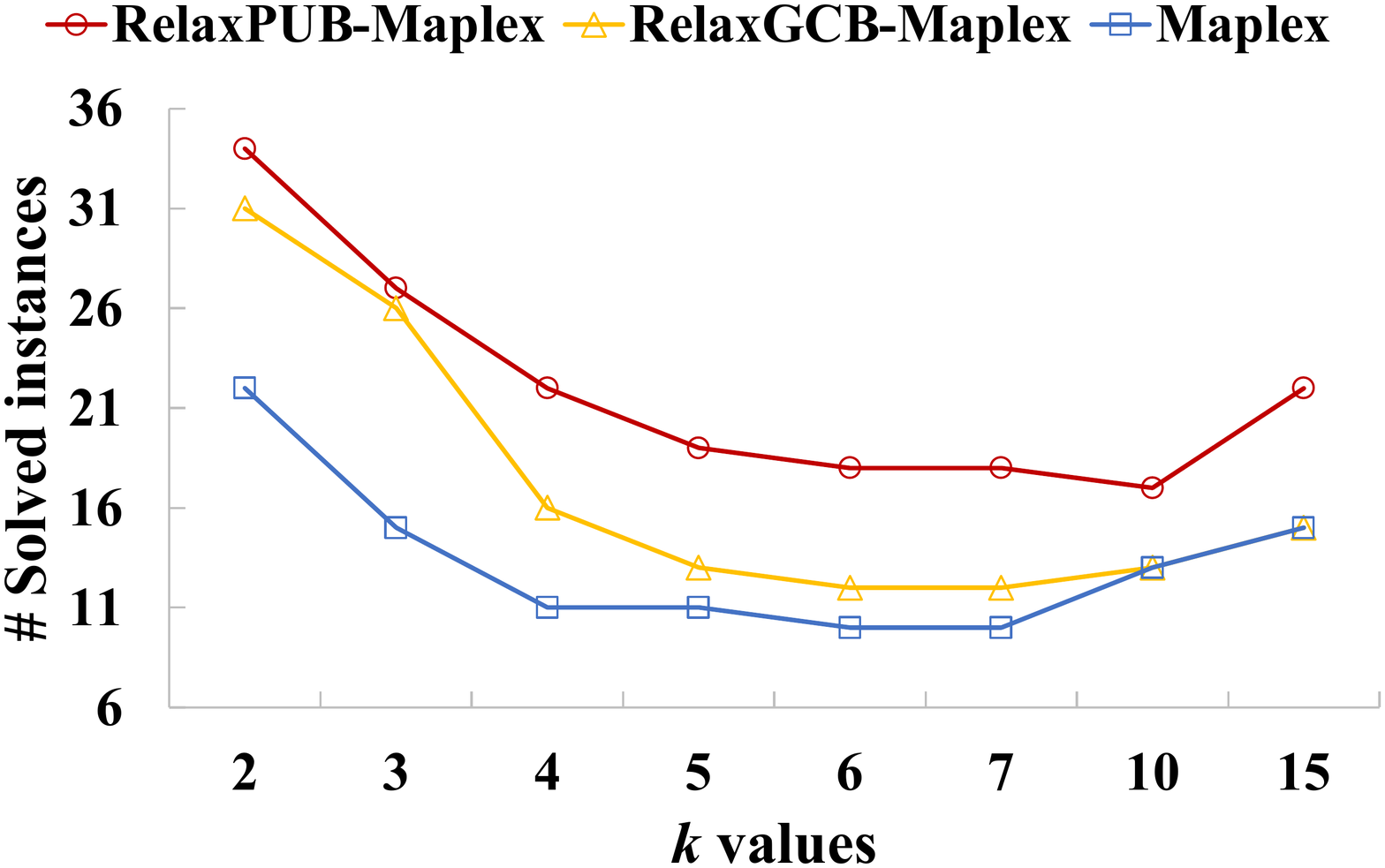}
\label{fig-Maplex-DIAMCS2}}
\subfigure[Comparison with kPlexS]{
\includegraphics[width=0.49\columnwidth]{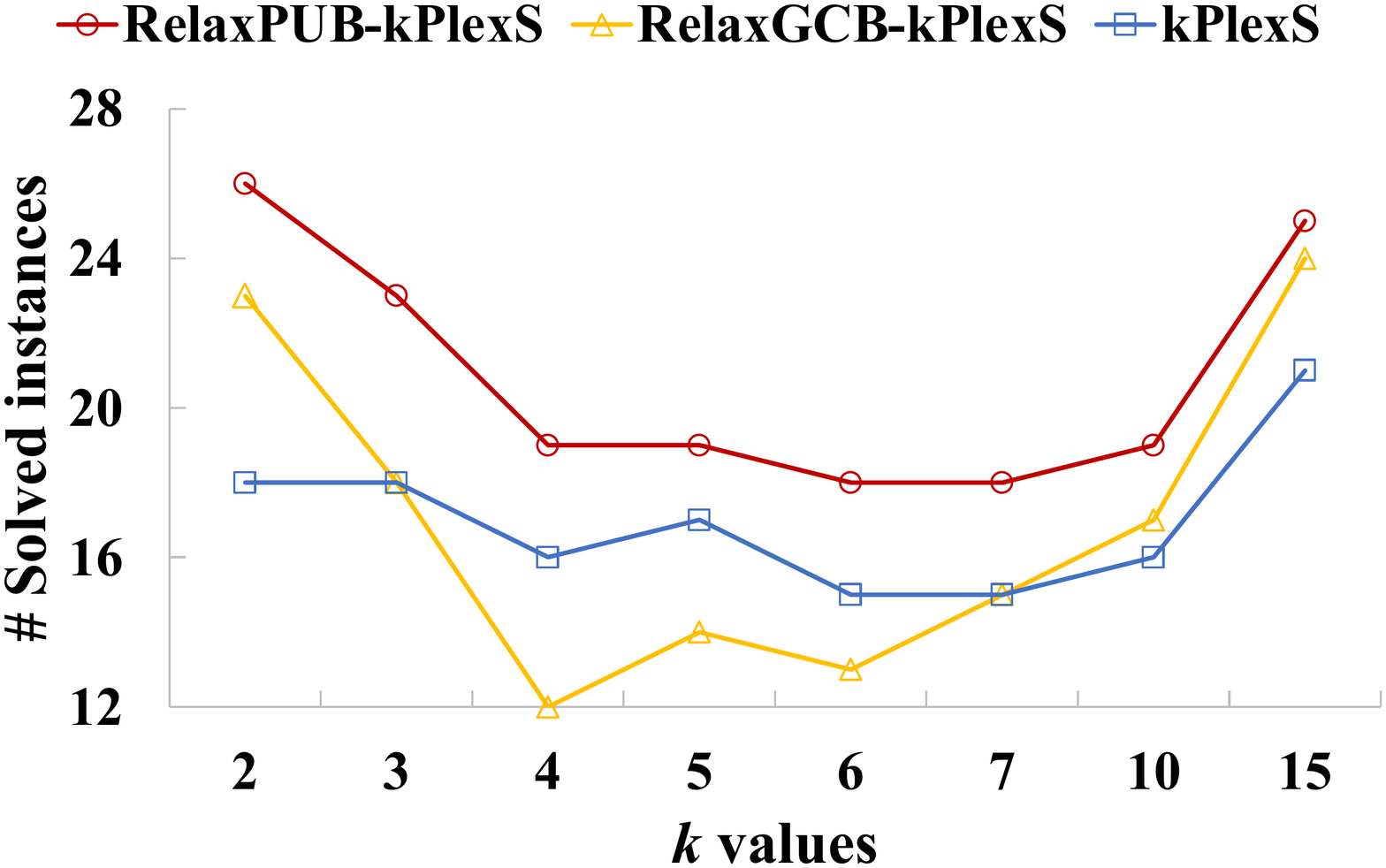}
\label{fig-kPlexS-DIAMCS2}}
\subfigure[Comparisons with DiseMKP]{
\includegraphics[width=0.49\columnwidth]{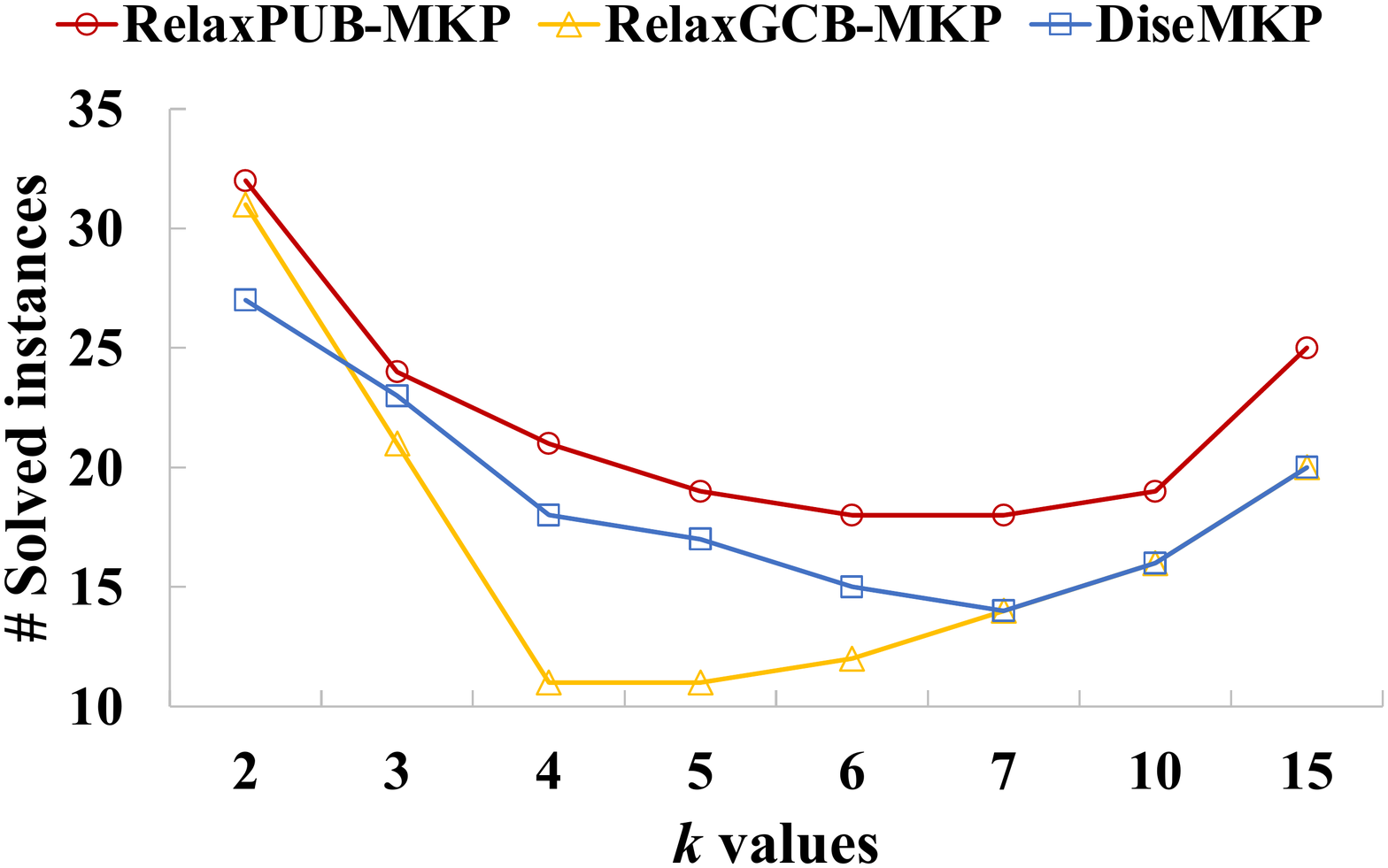}
\label{fig-DiseMKP-DIAMCS2}}
\subfigure[Comparison with KPLEX]{
\includegraphics[width=0.49\columnwidth]{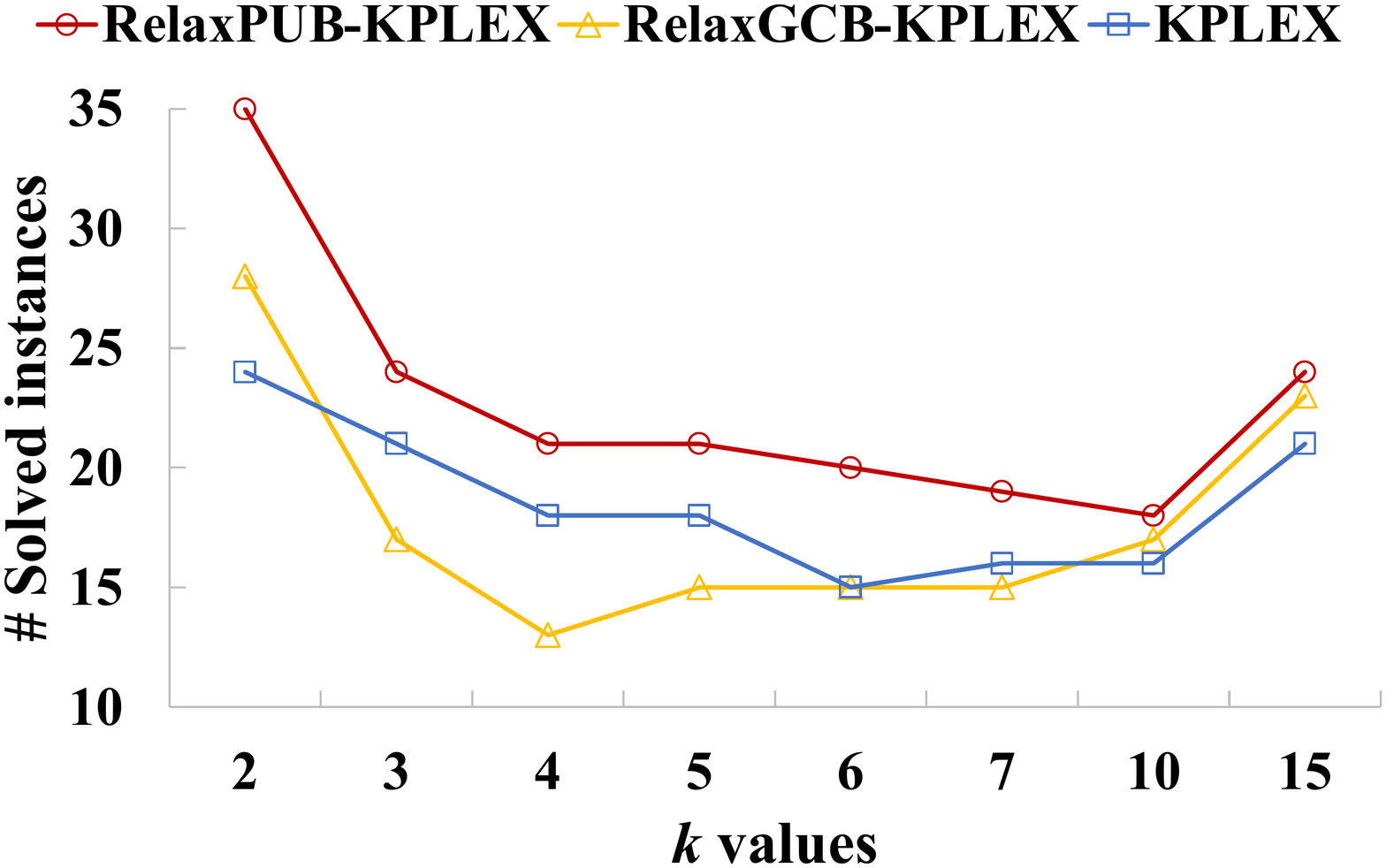}
\label{fig-KPLEX-DIAMCS2}}
\vspace{-1em}\caption{Comparisons on the dense 2nd DIMACS benchmark.
For the baselines, Maplex is based on GCB while the other three on PUB.
\vspace{-1em}}
\label{fig-DIAMCS2}
\end{figure*}

\begin{figure*}[!t]
\centering
\subfigure[Comparison with Maplex]{
\includegraphics[width=0.49\columnwidth]{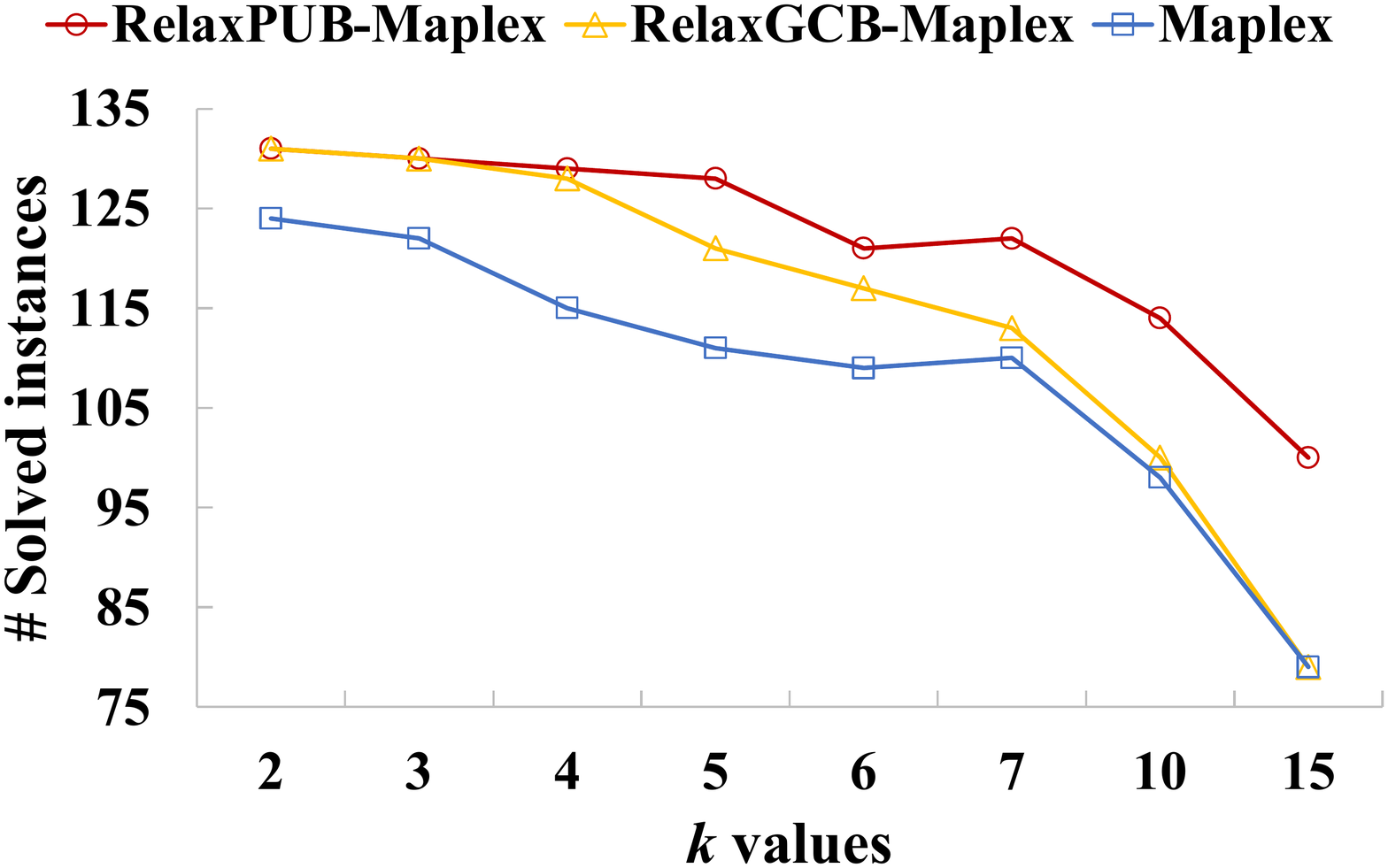}
\label{fig-Maplex-RealWorld}}
\subfigure[Comparison with kPlexS]{
\includegraphics[width=0.49\columnwidth]{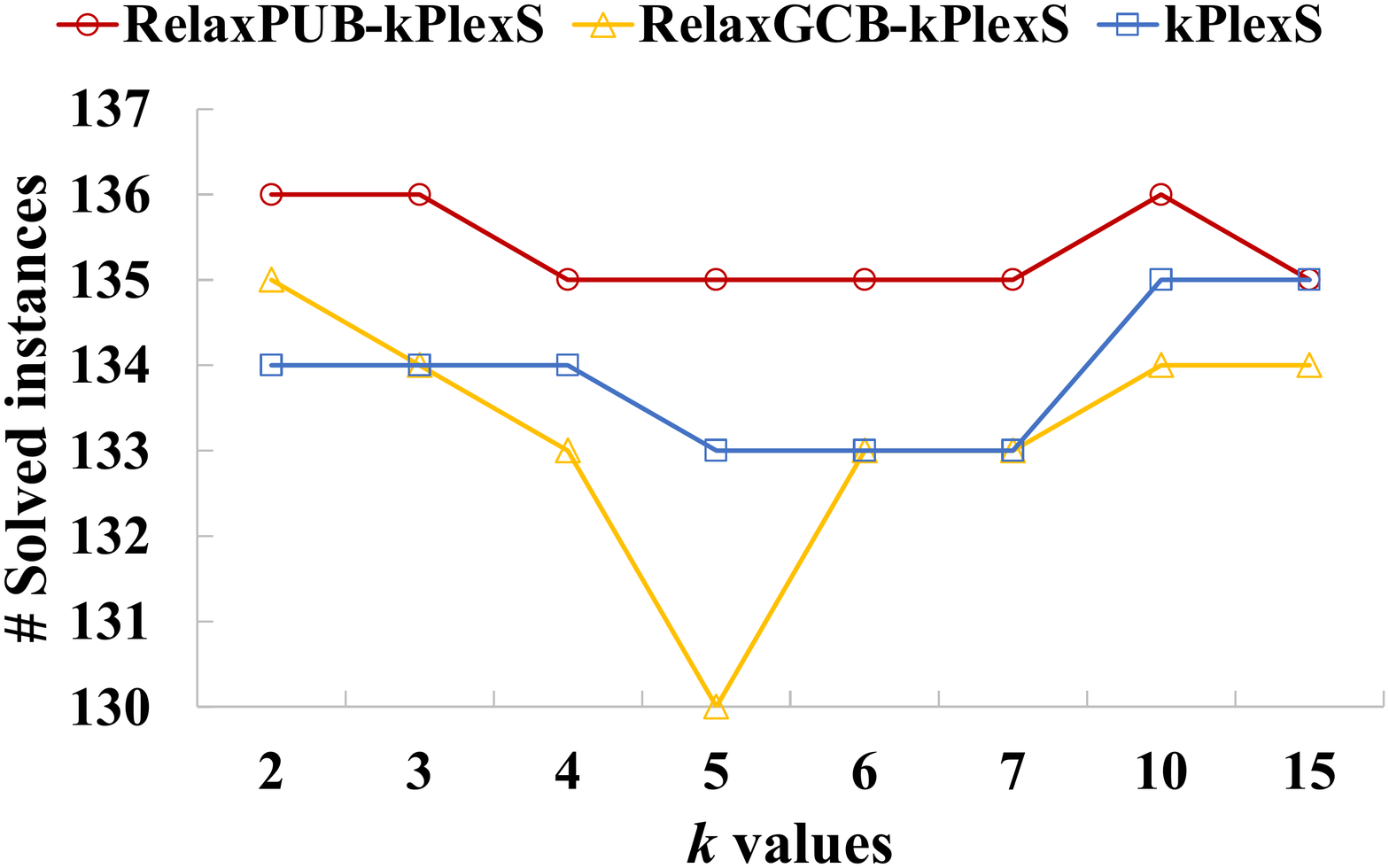}
\label{fig-kPlexS-RealWorld}}
\subfigure[Comparison with DiseMKP]{
\includegraphics[width=0.49\columnwidth]{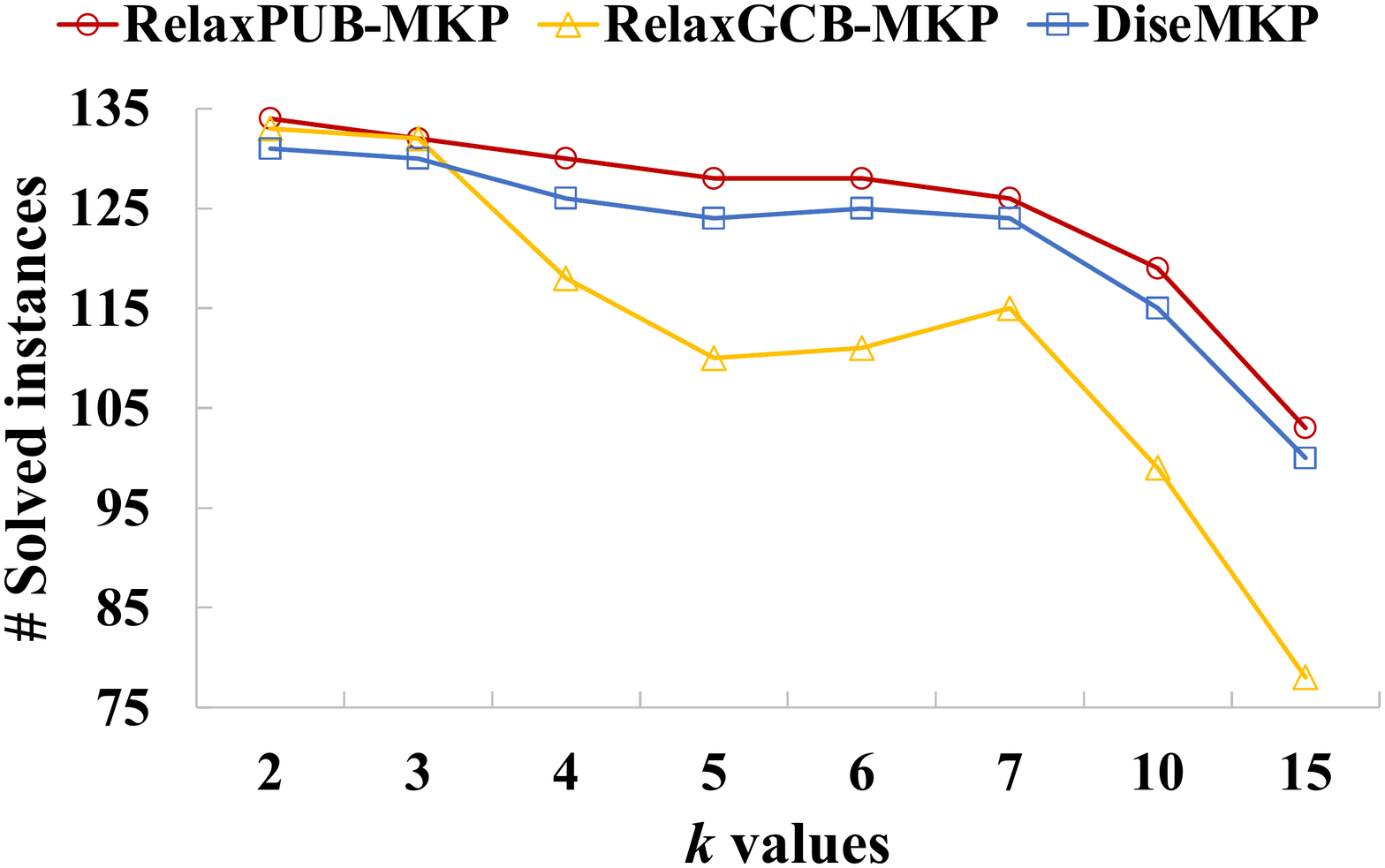}
\label{fig-DiseMKP-RealWorld}}
\subfigure[Comparison with KPLEX\vspace{-1em}]{
\includegraphics[width=0.49\columnwidth]{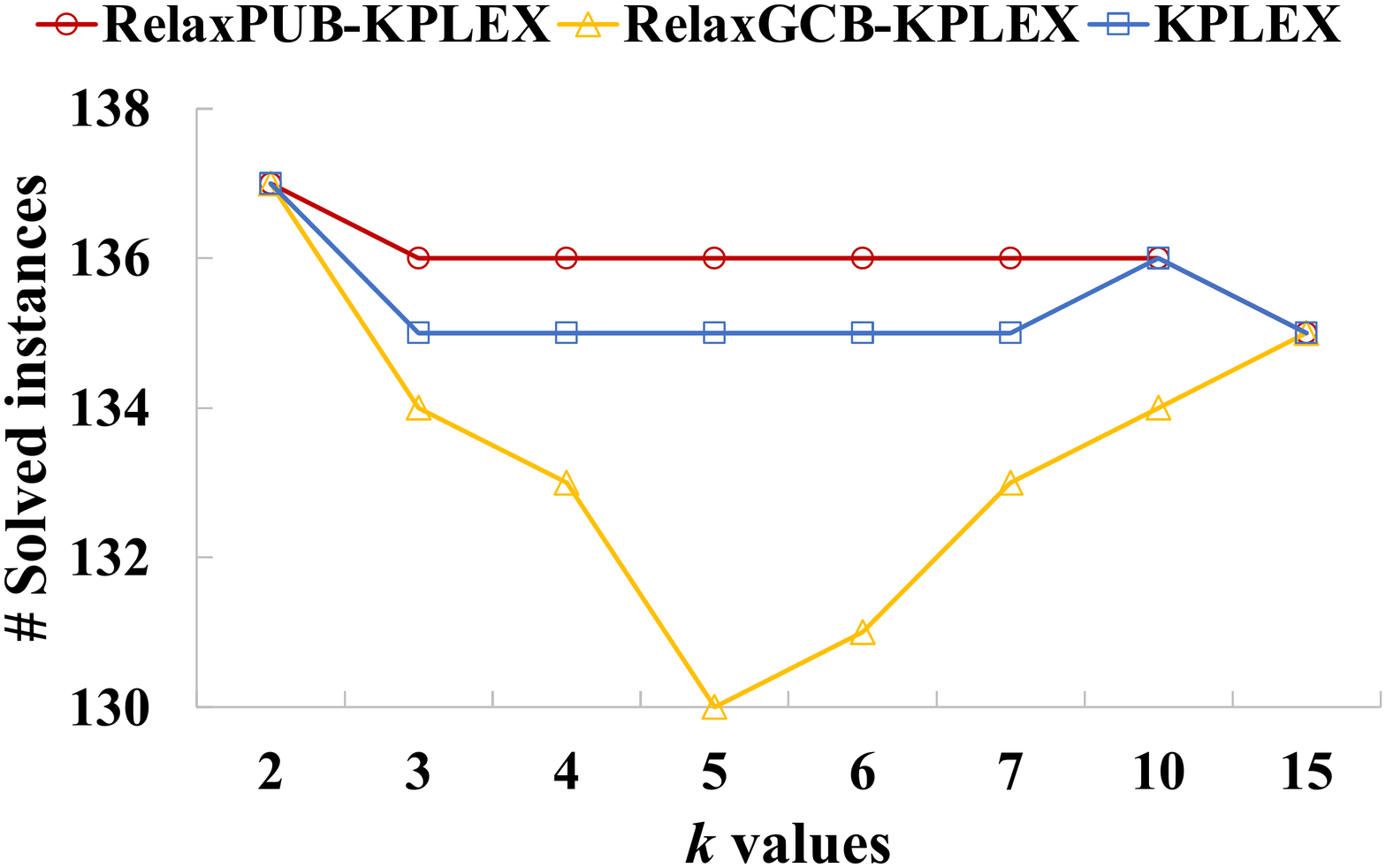}
\label{fig-KPLEX-RealWorld}}
\vspace{-1em}\caption{Comparisons on the sparse Real-world benchmark.
For the baselines, Maplex is based on GCB 
while the other three on PUB
\vspace{-0.5em}}
\label{fig-RealWorld}
\end{figure*}

\begin{table*}[!t]
\centering
\footnotesize
\resizebox{\linewidth}{!}{
\begin{tabular}{c|l|lllll|lllll|lllll|lllll} \toprule
\multicolumn{1}{c|}{\multirow{2}{*}{$k$}} & \multicolumn{1}{l|}{\multirow{2}{*}{\hspace{-0.3em}Instance}} & \multicolumn{3}{c}{RelaxPUB-Maplex}      & \multicolumn{2}{c|}{Maplex} & \multicolumn{3}{c}{RelaxPUB-kPlexS}      & \multicolumn{2}{c|}{kPlexS} & \multicolumn{3}{c}{RelaxPUB-MKP}          & \multicolumn{2}{c|}{DiseMKP} & \multicolumn{3}{c}{RelaxPUB-KPLEX}       & \multicolumn{2}{c}{KPLEX} \\
\multicolumn{1}{l|}{} & \multicolumn{1}{c|}{}                          & \textit{Tree}           & \textit{Time}           & \textit{Percent}  & \textit{Tree}         & \textit{Time}         & \textit{Tree}           & \textit{Time}           & \textit{Percent}  & \textit{Tree}         & \textit{Time}         & \textit{Tree}           & \textit{Time}           & \textit{Percent}   & \textit{Tree}          & \textit{Time}         & \textit{Tree}           & \textit{Time}           & \textit{Percent}  & \textit{Tree}        & \textit{Time}        \\ \hline
\multirow{9}{*}{2}  & \hspace{-0.3em}brock200-3                                     & \textbf{9.432} & \textbf{10.06} & 49.7\% & 790.8        & 107.8        & \textbf{8.343} & \textbf{292.0} & 49.9\% & 52.35        & 1,615        & \textbf{7.393} & \textbf{21.93} & 46.2\% & 85.88         & 26.06        & \textbf{27.43} & \textbf{34.63} & 79.8\% & 225.2       & 80.52       \\
                      & \hspace{-0.3em}brock200-4                                     & \textbf{24.99} & \textbf{30.07} & 50.2\% & 4015         & 576.9        & \textbf{23.84} & \textbf{695.8} & 49.7\% & -            & -            & \textbf{12.60} & \textbf{47.20} & 49.8\% & 279.6         & 103.1        & \textbf{68.81} & \textbf{84.81} & 78.8\% & 780.2       & 250.0       \\
                      & \hspace{-0.3em}C125.9                                         & \textbf{88.36} & \textbf{186.3} & 66.3\% & -            & -            & \textbf{24.56} & \textbf{162.8} & 71.8\% & -            & -            & \textbf{15.32} & \textbf{69.65} & 70.5\% & -             & -            & \textbf{29.61} & \textbf{37.93} & 88.0\% & -           & -           \\
                      & \hspace{-0.3em}keller4                                        & \textbf{3.715} & \textbf{4.015} & 46.1\% & 1273         & 182.3        & \textbf{3.961} & \textbf{87.18} & 44.3\% & 92.45        & 1,238        & \textbf{2.666} & \textbf{8.016} & 42.0\% & 69.11         & 21.15        & \textbf{19.29} & \textbf{17.38} & 79.6\% & 562.6       & 105.3       \\
                      & \hspace{-0.3em}san200-0-9-1                                   & \textbf{0.004} & \textbf{0.051} & 93.0\% & -            & -            & \textbf{0.024} & \textbf{0.660} & 94.8\% & -            & -            & \textbf{0.003} & \textbf{0.105} & 95.2\% & -             & -            & \textbf{0.057} & \textbf{0.619} & 98.4\% & 64.21       & 27.33       \\
                      & \hspace{-0.3em}sanr200-0-7                                    & \textbf{96.72} & \textbf{135.8} & 49.5\% & -            & -            & \textbf{66.05} & \textbf{1,656} & 51.5\% & -            & -            & \textbf{40.59} & \textbf{153.0} & 49.9\% & 1130          & 475.4        & \textbf{171.7} & \textbf{206.6} & 79.5\% & 2681        & 752.0       \\
                      & \hspace{-0.3em}socfb-Duke14                                   & \textbf{0.579} & \textbf{2.313} & 78.7\% & 195.6        & 45.05        & \textbf{0.110} & \textbf{4.397} & 85.8\% & 1.046        & 36.58        & \textbf{0.213} & \textbf{2.121} & 83.4\% & 243.6         & 154.5        & \textbf{0.281} & \textbf{2.403} & 94.3\% & 1.598       & 2.957       \\
                      & \hspace{-0.3em}socfb-UF                                       & \textbf{0.253} & \textbf{3.217} & 88.3\% & -            & -            & \textbf{0.094} & \textbf{2.012} & 93.2\% & 0.190        & 3.310        & \textbf{0.069} & \textbf{2.800} & 91.5\% & 413.1         & 299.3        & \textbf{0.107} & \textbf{1.741} & 91.9\% & 0.234       & 1.850       \\
                      & \hspace{-0.3em}socfb-Uillinois                                & \textbf{0.229} & \textbf{7.991} & 89.3\% & -            & -            & \textbf{0.022} & \textbf{2.224} & 92.8\% & 0.023        & 2.628        & \textbf{0.168} & \textbf{4.138} & 81.0\% & 18.33         & 13.05        & \textbf{0.024} & \textbf{1.774} & 85.8\% & 0.027       & 1.845       \\ \hline
\multirow{9}{*}{3}  & \hspace{-0.3em}hamming6-2                                     & \textbf{393.2} & \textbf{399.6} & 57.6\% & -            & -            & \textbf{204.4} & \textbf{349.2} & 49.8\% & -            & -            & \textbf{137.4} & \textbf{234.1} & 49.1\% & 826.6         & 304.6        & \textbf{197.0} & \textbf{136.2} & 60.2\% & 1157        & 203.3       \\
                      & \hspace{-0.3em}MANN-a81                                       & \textbf{0.001} & \textbf{0.001} & 100\%  & -            & -            & \textbf{0.001} & \textbf{534.9} & 96.3\% & 0.001        & 556.1        & \textbf{0.001} & \textbf{62.67} & 98.3\% & 0.003         & 63.61        & \textbf{0.004} & \textbf{568.4} & 98.9\% & 0.406       & 605.2       \\
                      & \hspace{-0.3em}p-hat300-2                                     & \textbf{174.7} & \textbf{340.9} & 57.0\% & -            & -            & \textbf{123.6} & \textbf{1,565} & 57.2\% & -            & -            & \textbf{36.99} & \textbf{200.0} & 59.2\% & -             & -            & \textbf{401.9} & \textbf{470.9} & 75.3\% & 1293        & 529.0       \\
                      & \hspace{-0.3em}socfb-UF                                       & \textbf{82.41} & \textbf{254.3} & 87.3\% & -            & -            & \textbf{0.094} & \textbf{1.975} & 82.1\% & 0.324        & 4.172        & \textbf{6.950} & \textbf{40.60} & 86.7\% & 440.6         & 413.4        & \textbf{0.067} & \textbf{1.877} & 89.4\% & 0.079       & 2.104       \\
                      & \hspace{-0.3em}socfb-Indiana                                  & \textbf{2.437} & \textbf{8.069} & 89.1\% & 1002         & 308.6        & \textbf{0.007} & \textbf{1.722} & 82.8\% & 0.008        & 1.879        & \textbf{1.110} & \textbf{7.302} & 84.9\% & 587.8         & 391.5        & \textbf{0.006} & \textbf{1.436} & 89.9\% & 0.006       & 1.708       \\
                      & \hspace{-0.3em}soc-flixster                                   & \textbf{8.732} & \textbf{20.00} & 74.6\% & -            & -            & \textbf{0.725} & \textbf{8.152} & 73.5\% & 7.394        & 110.5        & \textbf{2.084} & \textbf{13.15} & 76.2\% & 118.0         & 85.04        & \textbf{0.280} & \textbf{3.982} & 84.7\% & 0.470       & 5.674       \\
                      & \hspace{-0.3em}soc-lastfm                                     & \textbf{1.384} & \textbf{6.018} & 54.8\% & 78.33        & 29.91        & \textbf{0.327} & \textbf{17.86} & 60.2\% & 0.785        & 43.87        & \textbf{0.729} & \textbf{9.750} & 52.4\% & 7.163         & 10.26        & \textbf{1.549} & \textbf{53.88} & 77.5\% & 2.819       & 95.76       \\
                      & \hspace{-0.3em}soc-slashdot                                   & \textbf{1.607} & \textbf{1.922} & 74.2\% & 2461         & 289.2        & \textbf{0.095} & \textbf{0.715} & 72.3\% & 0.299        & 2.825        & \textbf{0.245} & \textbf{0.910} & 76.9\% & 7.847         & 4.577        & \textbf{0.096} & \textbf{0.829} & 82.4\% & 0.131       & 0.766       \\
                      & \hspace{-0.3em}tech-WHOIS                                     & \textbf{24.96} & \textbf{72.23} & 85.3\% & -            & -            & \textbf{0.049} & \textbf{0.451} & 84.4\% & 0.134        & 1.647        & \textbf{0.384} & \textbf{2.514} & 89.8\% & 6.996         & 6.043        & \textbf{0.028} & \textbf{0.279} & 90.9\% & 0.034       & 0.483       \\ \hline
\multirow{6}{*}{6}  & \hspace{-0.3em}c-fat200-1                                     & \textbf{0.001} & \textbf{0.001} & 92.8\% & 0.003        & 0.001        & \textbf{0.001} & \textbf{0.001} & 80.6\% & 0.001        & 0.002        & \textbf{0.002} & \textbf{0.014} & 82.9\% & 0.002         & 0.018        & \textbf{0.001} & \textbf{0.001} & 66.0\% & 0.001       & 0.001       \\
                      & \hspace{-0.3em}san200-0-7-1                                   & \textbf{0.001} & \textbf{0.017} & 95.9\% & -            & -            & \textbf{0.001} & \textbf{0.040} & 93.3\% & 3.163        & 4.928        & \textbf{0.001} & \textbf{0.037} & 93.7\% & 0.329         & 0.177        & \textbf{0.001} & \textbf{0.043}          & 87.1\% & -           & -           \\
                      & \hspace{-0.3em}san200-0-7-2                                   & \textbf{0.001} & \textbf{0.012} & 80.1\% & 11.99        & 6.417        & \textbf{0.004} & \textbf{0.156} & 70.4\% & -            & -            & \textbf{0.005} & \textbf{0.066} & 68.8\% & -             & -            & \textbf{0.002} & \textbf{0.122}          & 81.8\% & -           & -           \\
                      & \hspace{-0.3em}socfb-Berkeley13                               & \textbf{0.078} & \textbf{1.514} & 94.4\% & 1780         & 249.5        & \textbf{0.001} & \textbf{0.890} & 50.0\% & 0.001        & 0.905        & \textbf{0.244} & \textbf{1.890} & 90.4\% & 4.509         & 4.586        & \textbf{0.001} & \textbf{0.828}          & 28.2\% & 0.001       & 0.867       \\
                      & \hspace{-0.3em}socfb-MIT                                      & \textbf{0.189} & \textbf{0.591} & 78.4\% & 15876        & 1684         & \textbf{0.002} & \textbf{0.317} & 70.7\% & 0.002        & 0.319        & \textbf{0.041} & \textbf{0.523} & 91.4\% & 2.220         & 1.764        & \textbf{0.002} & \textbf{0.271}          & 57.0\% & 0.002       & 0.340       \\
                      & \hspace{-0.3em}soc-gowalla                                    & \textbf{12.45} & \textbf{11.61} & 54.2\% & 779.5        & 120.1        & \textbf{0.003} & \textbf{0.430} & 63.4\% & 0.005        & 0.572        & \textbf{6.042} & \textbf{16.65} & 53.1\% & 59.68         & 43.12        & \textbf{0.004} & \textbf{0.199}          & 69.3\% & 0.004       & 0.242       \\ \hline
\multirow{3}{*}{10} & \hspace{-0.3em}bio-dmela                                      & \textbf{1.245} & \textbf{14.64} & 39.1\% & -            & -            & \textbf{0.004} & \textbf{0.067} & 48.7\% & 0.007        & 0.076        & \textbf{0.085} & \textbf{0.113} & 30.9\% & 4.615         & 0.221        & \textbf{0.019} & \textbf{0.037} & 57.1\% & 0.030       & 0.038       \\
                      & \hspace{-0.3em}ia-enron-large                                 & \textbf{5.779} & \textbf{4.850} & 55.4\% & 156.3        & 24.10        & \textbf{0.001} & \textbf{0.103} & 76.3\% & 0.001        & 0.115        & \textbf{14.80} & \textbf{24.89} & 39.6\% & 266.7         & 178.1        & \textbf{0.002} & \textbf{0.064} & 72.1\% & 0.002       & 0.065       \\
                      & \hspace{-0.3em}tech-RL-caida                                  & \textbf{0.662} & \textbf{0.803} & 46.9\% & 18.52        & 4.693        & \textbf{0.001} & \textbf{0.251} & 65.9\% & 0.001        & 0.258        & \textbf{0.017} & \textbf{0.363} & 73.2\% & 1.160         & 0.500        & \textbf{0.001} & \textbf{0.067} & 83.3\% & 0.001       & 0.068       \\ \hline
\multirow{3}{*}{15} & \hspace{-0.3em}C125-9                                         & \textbf{152.7} & \textbf{507.5} & 79.6\% & -            & -            & \textbf{5.460} & \textbf{30.22} & 70.4\% & 12.21        & 96.15        & \textbf{2.712} & \textbf{16.10} & 82.6\% & 7.771         & 22.62        & \textbf{1.194} & \textbf{5.414} & 83.5\% & 17.47       & 40.32       \\
                      & \hspace{-0.3em}bio-diseasome                                  & \textbf{0.234} & \textbf{0.145} & 80.1\% & 1011         & 169.0        & \textbf{0.001} & \textbf{0.001} & 81.3\% & 0.001        & 0.001        & \textbf{0.055} & \textbf{0.024} & 57.7\% & 80.71         & 2.321        & \textbf{0.000} & \textbf{0.000} & -      & 0.000       & 0.001       \\
                      & \hspace{-0.3em}socfb-uci-uni                                  & \textbf{22.61} & \textbf{107.5} & 47.9\% & -            & -            & \textbf{0.019} & \textbf{49.43} & 58.1\% & 0.517        & 53.96        & \textbf{26.93} & \textbf{295.4} & 35.6\% & -             & -            & \textbf{0.185} & \textbf{6.775} & 62.7\% & 1.096       & 9.322       \\ \bottomrule
\end{tabular}}
\caption{Comparison on 30 representative MKP instances with $k = 2,3,6,10,15$. The search tree size is in $10^5$, and the time is in seconds. The percent indicates the percentage of the number of times RelaxGCB is used in RelaxPUB. Better results appear in bold.\vspace{-1em}}
\label{table-Instances}
\end{table*}

\subsection{Experimental Setup}
All the algorithms were implemented in C++ and run on a server using an AMD EPYC 7H12 CPU, running Ubuntu 18.04 Linux operation system. We test the algorithms on two public benchmarks that are widely used in the literature of the baselines, the 2nd DIMACS benchmark\footnote{http://archive.dimacs.rutgers.edu/pub/challenge/graph/\\benchmarks/clique/} that contains 80 (almost dense) graphs with up to 4,000 vertices and densities ranging from 0.03 to 0.99, and the Real-world benchmark\footnote{http://lcs.ios.ac.cn/\%7Ecaisw/Resource/realworld\%20\\graphs.tar.gz} that contains 139 real-world sparse graphs from the Network Data Repository~\cite{RA15}.

We choose the two sets of benchmarks because the 2nd DIMACS benchmark is also widely used to evaluate MCP, one of the most closely related problems to MKP, and the Real-world benchmark is widely used for analyzing various complex networks, one of the most important application areas of MKP. Moreover, the structures of the two benchmarks are distinct, helping evaluate the robustness of the algorithms.

For each graph, we generate 8 MKP instances with $k \in \{2,3,4,5,6,7,10,15\}$, and set the cut-off time to 1,800 seconds per instance, 
following the settings of the baselines.

\begin{figure*}[!t]
\centering
\subfigure[On Maplex]{
\includegraphics[width=0.78\columnwidth]{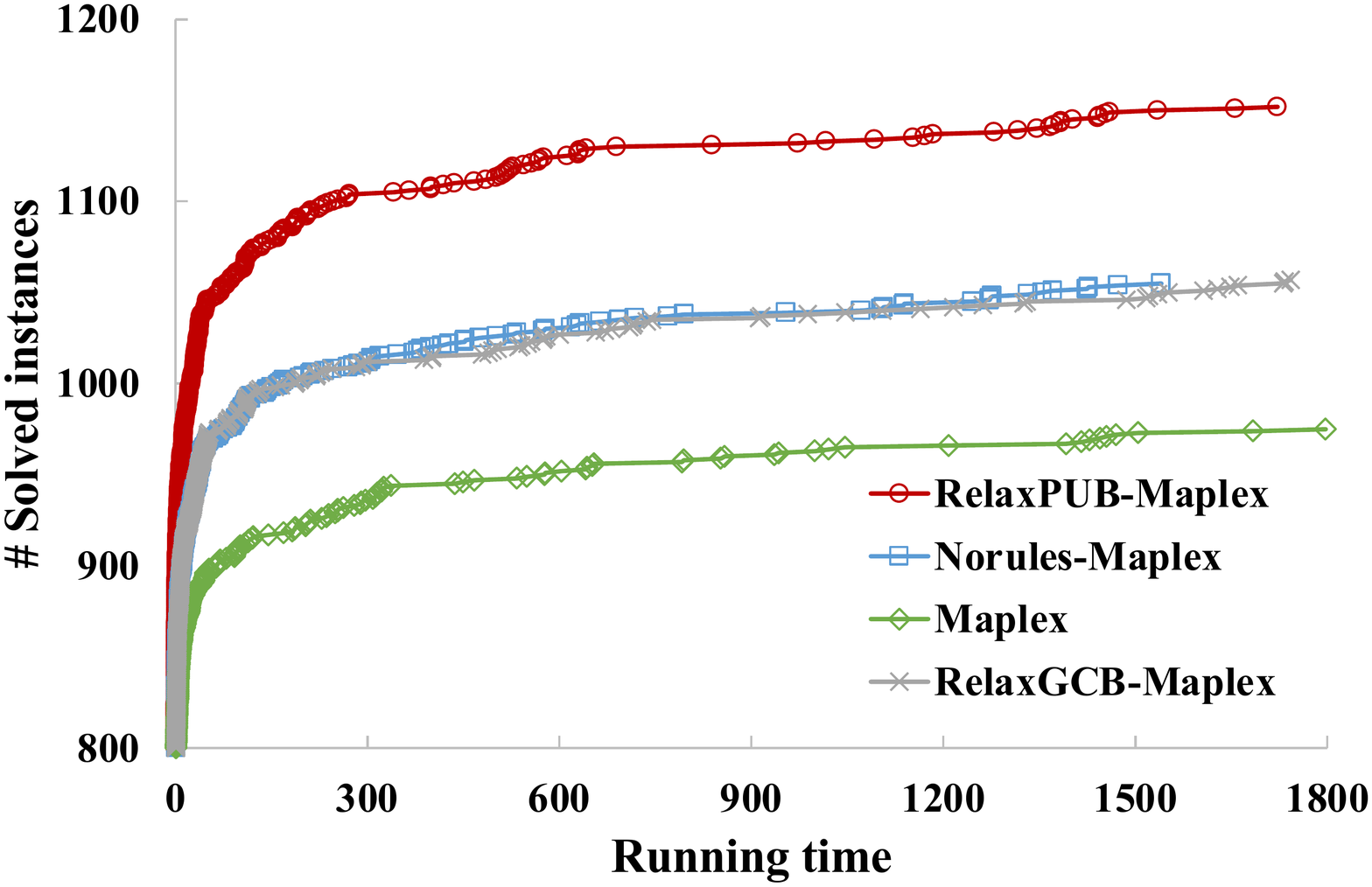}
\label{fig-Maplex-Ablation}}~~~~
\subfigure[On kPlexS]{
\includegraphics[width=0.78\columnwidth]{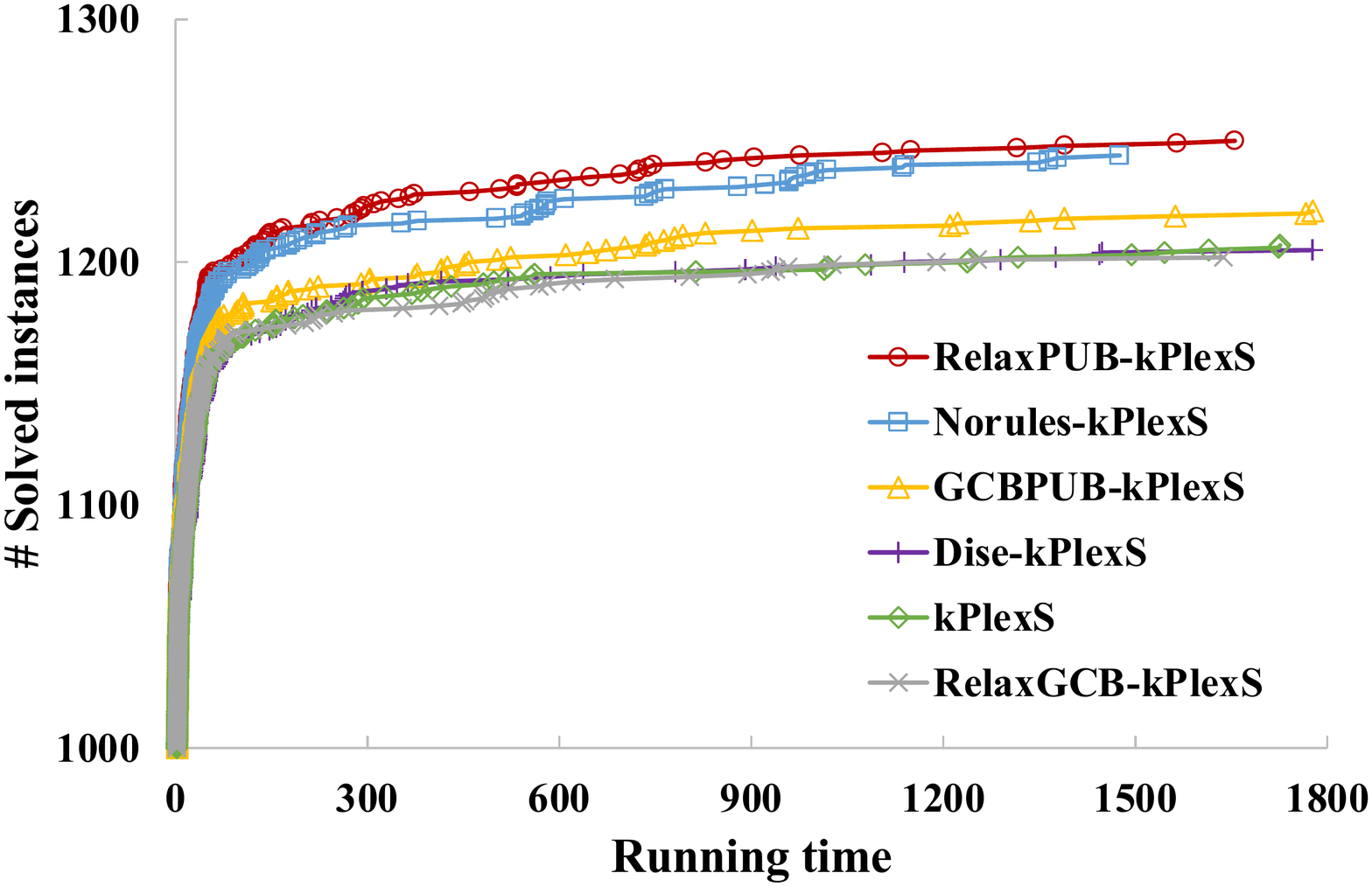}
\label{fig-kPlexS-Ablation}}
\subfigure[On DiseMKP]{
\includegraphics[width=0.78\columnwidth]{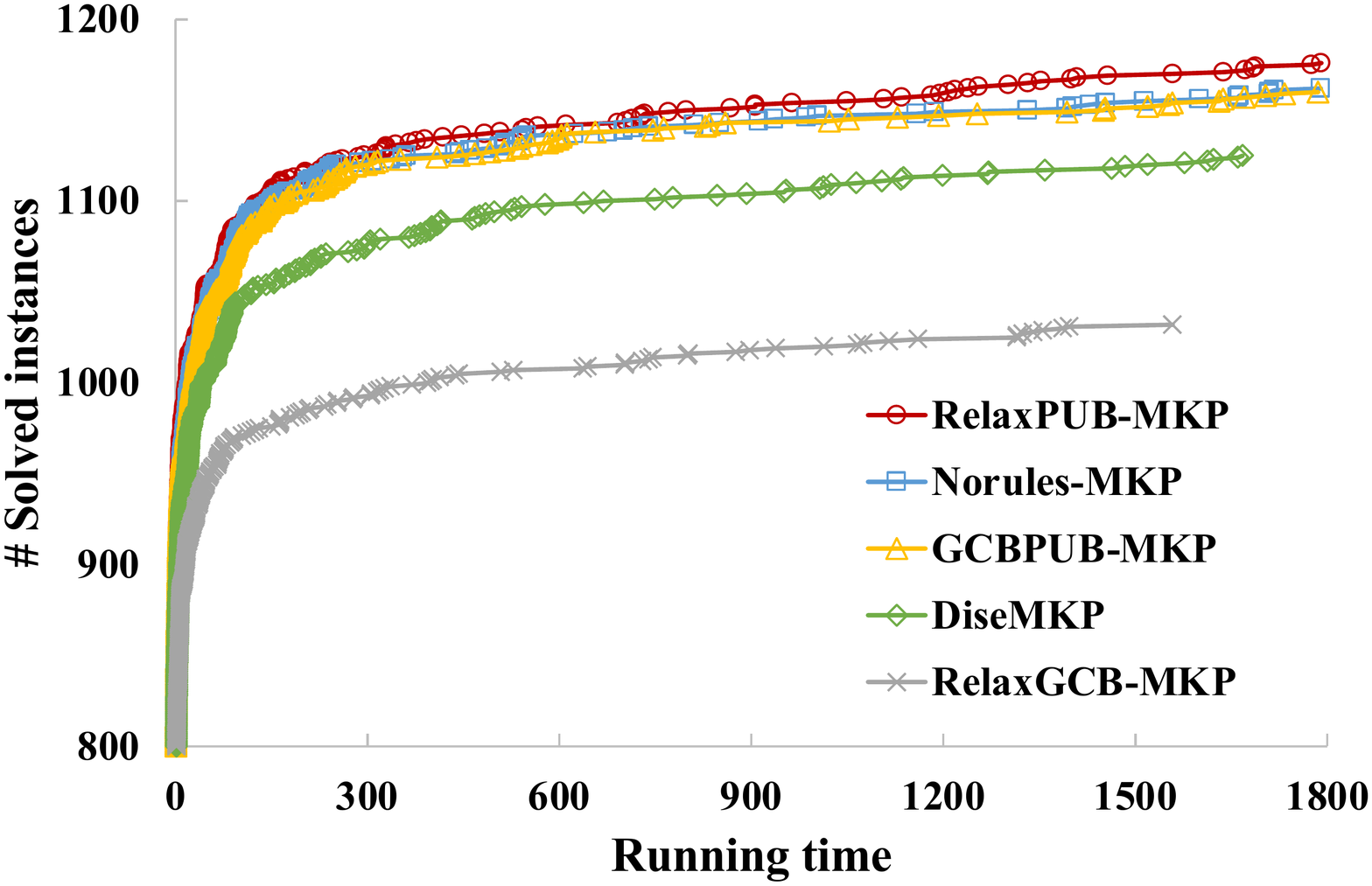}
\label{fig-DiseMKP-Ablation}}~~~~
\subfigure[On KPLEX]{
\includegraphics[width=0.78\columnwidth]{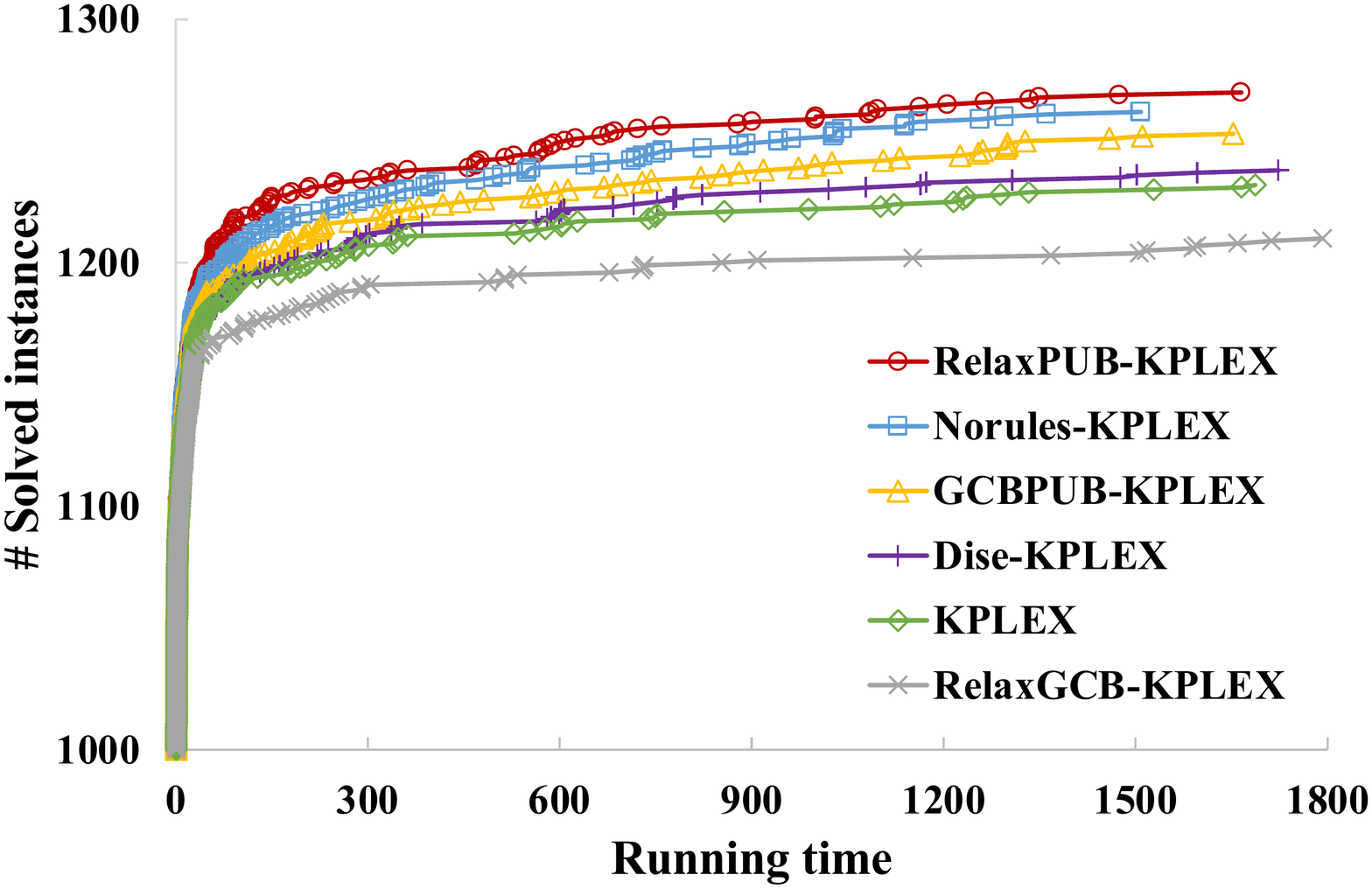}
\label{fig-KPLEX-Ablation}}
\vspace{-1em}\caption{Ablation studies on each baseline over all the tested instances.\vspace{-1em}}
\label{fig-Ablation}
\end{figure*}

\subsection{
Performance Evaluation}
The comparison results between the algorithms with our RelaxGCB and RelaxPUB bounds and the baselines in dense 2nd DIMACS and sparse Real-world benchmarks are summarized in Figures~\ref{fig-DIAMCS2} and~\ref{fig-RealWorld}, respectively. 
The results are expressed by the number of MKP instances solved by each algorithm within the cut-off time for different $k$ values. 
Note that Maplex only contains the GCB, and the other three baselines only contain the PUB. 
1) From (a) of the two figures, one can observe that our RelaxGCB significantly outperforms GCB. 
2) From (b) to (d)  of the two figures, one can observe that RelaxGCB is complementary to PUB. 
3) From all the figures, one can observe that our RelaxPUB makes full use of the complementarity of RelaxGCB and PUB, and significantly improves all the baselines in solving both dense and massive sparse graphs over diverse $k$ values, indicating its 
dominant performance over the state-of-the-art baselines, excellent generalization over different graphs, and strong robustness over diverse $k$ values.


With the increment of the $k$ values, the number of solved instances usually decays, because the number of vertices removed by graph reduction decreases accordingly, making the follow-up BnB calculation more complicated with a larger number of required branches. 
However, we notice that the number of solved instances increases for larger $k$ values (\eg, 10 and 15) on DIMACS2. This is because a $k$-plex with larger $k$-values can contain more vertices, and the DIMACS2 graphs are generally small and dense, making the most vertices in a graph contained. 


Following the convention of the baselines, we also present detailed results of the baselines and their improvements with RelaxPUB in solving 30 representative 2nd DIAMCS and Real-world instances with $k = 2,3,6,10,15$ in Table~\ref{table-Instances}. We report their running times in seconds (column \textit{Time}), the sizes of their entire search trees in $10^5$ (column \textit{Tree}) to solve the instances, and the percentage of the number of times RelaxGCB is selected and outperforms the DisePUB in RelaxPUB (column \textit{Percent}). Better results are highlighted in bold, and symbol `-' means the algorithm cannot solve the instance within the cut-off time.

The results show that for each pair of tested algorithms, our new upper bounds can help the baseline algorithm prune 
significantly more branches, reducing its search tree sizes by several orders of magnitude for instances that both can solve within the cut-off time. There are also many instances that the baseline algorithms cannot solve within the cut-off time, while the algorithms with our upper bounds can solve with few branches and much less calculation time. Moreover, we can observe that RelaxGCB contributes a lot in solving these instances, indicating again the complementarity of RelaxGCB and DisePUB. 




\subsection{Ablation Study}
In this subsection, we perform ablation studies to evaluate the effectiveness of the proposed TISUB and the two rules (see Lemmas~\ref{lemma-RelaxGCB1}, \ref{lemma-R1}, and \ref{lemma-R2})
in our proposed upper bounds. For the kPlexS, DiseMKP, and KPLEX baselines having the PUB, we generate a ``Norules'' variant, which uses our RelaxPUB without Rules 1 and 2, and a ``GCBPUB'' variant, which uses our RelaxPUB and replaces its RelaxGCB with the GCB in Maplex. 
Moreover, since the kPlexS and KPLEX algorithms use the previous PUB proposed in~\cite{JZX+21}, we apply the newest DisePUB to them and obtain two variants: Dise-kPlexS and Dise-KPLEX.
For the Maplex baseline that is only based on GCB, we generate a ``Norules'' variant, which uses our RelaxGCB without Rules 1 and 2.

We perform four groups of ablation studies based on each baseline over all the 1,752 instances, as summarized in Figure~\ref{fig-Ablation}. The results are expressed by the variation in the number of solved instances for each algorithm over the running time (in seconds). 
The results show that the ``GCBPUB'' variants are better than the baselines, indicating that combining coloring-based and partition-based upper bounds by the mechanism in RelaxPUB can make use of their complementarity. 
The ``Norules'' variants are better than the ``GCBPUB'' variants, indicating that TISUB is a significant improvement over GCB. The new algorithms with RelaxPUB are better than the ``Norules'' variants, indicating that our proposed two rules can further improve TISUB. Moreover, DisePUB can hardly improve kPlexS and KPLEX, indicating that the improvements of the RelaxPUB series over the baselines originate from RelaxGCB rather than using the newest DisePUB.


\section{Conclusion}
We proposed two new upper bounds for the Maximum $k$-plex Problem (MKP), termed RelaxGCB and RelaxPUB. RelaxGCB first tights the previous graph color bound (GCB) by considering the connectivity between vertices more thoroughly and relaxes the restrictive independent set structure by considering the relaxation property of MKP. 
RelaxPUB further combines RelaxGCB and an advanced partition-based upper bound in a novel way, making full use of their complementarity. 
We replaced the GCB in Maplex and the partition-based upper bounds in kPlexS, DiseMKP, and KPLEX with our two bounds, RelaxGCB and RelaxPUB, respectively, producing eight new BnB MKP algorithms. 
Experiments on both dense and sparse benchmark datasets show that RelaxGCB is a significant improvement over GCB, and RelaxPUB 
exhibits clearly priority over the baselines and exhibits excellent robustness over various $k$ values and high generalization capability over different graphs. 


\bibliographystyle{named}
\bibliography{ijcai24}

\end{document}